\documentclass[acmsmall]{acmart}
\usepackage{amsfonts, datetime, cancel, mathtools, graphicx, amsmath, amsthm, bm, amsbsy, breqn, accents, float, listings, courier, lscape, multirow, multicol, longtable, dsfont, bbold, soul, color, tikz, scalerel, empheq, paralist, subcaption}
\usetikzlibrary{shapes, arrows}
\usetikzlibrary{arrows,decorations.markings}
\usetikzlibrary{decorations.pathmorphing}
\allowdisplaybreaks
\newcommand{\R}{\mathbb{R}}

\newcommand{\bracketRound}[1]{\left(#1\right)}
\newcommand{\bracketSquare}[1]{\left[#1\right]}

\newcommand{\prob}{p}

\newtheorem{theorem}{Theorem}

\newcommand{\compPower}{x}
\newcommand{\compPowerSet}{X}

\newcommand{\pow}{d}
\newcommand{\powSet}{D}
\newcommand{\wot}{w}
\newcommand{\wotSet}{W}
\newcommand{\stateCardinality}{n}
\newcommand{\powCardinality}{m}

\newcommand{\utilityAgent}{u}

\newcommand{\totalAgents}{N}
\newcommand{\rate}[2]{R(#1,#2)}

\newcommand{\fairness}{f}

\definecolor{lavander}{cmyk}{0,0.48,0,0}
\definecolor{violet}{cmyk}{0.79,0.88,0,0}
\definecolor{burntorange}{cmyk}{0,0.52,1,0}
\definecolor{visible}{rgb}{1,1,1}
\definecolor{invisible}{rgb}{0.8,0.8,0.8}
\tikzstyle{transactions}=[draw,rectangle,white, color=black, text=black,minimum width=10pt,minimum height=10pt]
\tikzstyle{blk_diag_square}=[draw,rectangle, white, color=black, text=black, minimum width=10pt, minimum height=10pt, align=center, inner sep=7, rounded corners=4]

\AtBeginDocument{%
	\providecommand\BibTeX{{%
			\normalfont B\kern-0.5em{\scshape i\kern-0.25em b}\kern-0.8em\TeX}}}

\begin{document}
	
	\title{Controlling Transaction Rate in Tangle Ledger: A Principal Agent Problem Approach}
	
	\author{Anurag Gupta}
	\email{ag2589@cornell.edu}
	\author{Vikram Krishnamurthy}
	\authornotemark[1]
	\email{vikramk@cornell.edu}
	\affiliation{%
		\institution{Cornell University}
		\streetaddress{136, Hoy Road}
		\city{Ithaca}
		\state{New York}
		\country{USA}
		\postcode{14853}
	}

	\begin{abstract}
		Tangle is a distributed ledger technology that stores data as a directed acyclic graph (DAG). Unlike blockchain, Tangle does not require dedicated miners for its operation; this makes Tangle suitable for Internet of Things (IoT) applications. Distributed ledgers have a built-in transaction rate control mechanism to prevent congestion and spamming; this is typically achieved by increasing or decreasing the proof of work (PoW) difficulty level based on the number of users. Unfortunately, this simplistic mechanism gives an unfair advantage to users with high computing power. This paper proposes a principal-agent problem (PAP) framework from microeconomics to control the transaction rate in Tangle. With users as agents and the transaction rate controller as the principal, we design a truth-telling mechanism to assign PoW difficulty levels to agents as a function of their computing power. The solution of the PAP is achieved by compensating a higher PoW difficulty level with a larger weight/reputation for the transaction. The mechanism has two benefits, (1) the security of Tangle is increased as agents are incentivized to perform difficult PoW, and (2) the rate of new transactions is moderated in Tangle. The solution of PAP is obtained by solving a mixed-integer optimization problem. We show that the optimal solution of the PAP increases with the computing power of agents. The structural results reduce the search space of the mixed-integer program and enable efficient computation of the optimal mechanism. Finally, via numerical examples, we illustrate the transaction rate control mechanism and study its impact on the dynamics of Tangle.
	\end{abstract}

	\ccsdesc[500]{Information systems~Distributed storage}
	\ccsdesc[500]{Networks~Network protocol design}
	\keywords{Distributed ledger, directed acyclic graph, Tangle, transaction rate control, proof of work (PoW), weight of transaction (WoT), principal-agent problem (PAP), mixed-integer optimization, linear program.}

	\maketitle
	
	\section{Introduction}
	\label{sec:introduction}
	Tangle was created by the IOTA foundation to enable a scalable distributed ledger technology for IoT applications with no transaction fees~\cite{2018:SP}. Tangle is an example of a directed acyclic graph (DAG) based distributed ledgers~\cite{2018-FB-et-al},~\cite{2020-YL-et-al},~\cite{2020:SP-et-al}. Tangle is a generalization of linear blockchain technology: it stores data as a DAG in contrast to the linear data structure of blockchain. The DAG structure of Tangle allows multiple transactions to be added to the ledger simultaneously. Hence, the throughput of Tangle is significantly higher than blockchain. The modified data structure also supports higher scalability and decreased mining cost~\cite{2018:SP}, which are essential for Internet of Things (IoT)~\cite{2014:LX-et-al},~\cite{2019-SK-et-al} applications. This is because IoT devices have limited computing power and can not afford high transaction fees associated with mining in blockchain. Therefore, Tangle is used in IoT applications such as automating real-time trade and exchange of renewable energy amongst neighbourhoods~\cite{2020-MZ-et-al}.
	
	There are two important elements of Tangle: proof of work (PoW) and weight of transactions (WoT)~\cite{2018:SP}. PoW requires users to solve a hash puzzle to add a transaction to the distributed ledger. Increasing the PoW difficulty level increases the computational complexity of the hash puzzle; this is used to control the rate of new transactions and increase the security of Tangle. WoT is the reputation assigned to the transaction. In Tangle, users must select two other transactions randomly and approve them to add a new transaction: if the two transactions do not conflict, they are approved. The probability that a transaction is selected for approval is proportional to its WoT. Therefore, users prefer a higher WoT for their transactions. In the current implementation of Tangle, WoT is a fixed function of PoW~\cite{2018:SP}; this makes Tangle susceptible to being dominated by users with powerful computational resources. Users with high computational resources can add a very high number of transactions with a high WoT. A natural question is \textit{how to make Tangle fair\footnote{By fair, we mean that every agent, irrespective of its computing power, can add transactions into Tangle at a low cost. Also, the mechanism should compensate agents for completing a difficult PoW. Our transaction rate control mechanism compensates agents' PoW with WoT. In Sec.\ref{sec:simulations}, we consider a simple fairness measure to trade off PoW and WoT. It is of interest in future work to incorporate schemes such as proportional,  max-min and social welfare fairness.} to all users?}
	
	In this paper, we utilize PoW difficulty level in conjunction with WoT to control the rate of new transactions in Tangle. Specifically, we formulate the transaction rate control problem for Tangle as a \textit{principal-agent problem (PAP) with adverse selection}~\cite{1995:AM}. The PAP has been studied widely in microeconomics to design a contract between two strategic players with misaligned utilities. Examples include labor contracts~\cite{2005:GM}, insurance market~\cite{2003:MH}, and  differential privacy~\cite{2014:CD-AR}. There are two types of PAP~\cite{1995:AM}: moral hazard and adversarial selection. We restrict our attention to the PAP with adverse selection as the underlying information asymmetry\footnote{In the PAP with adverse selection, the principal cannot observe the state of agents, and hence, it has to incentivize agents to reveal their state truthfully. The principal then assigns the effort level and compensation based on the revealed information to maximize its own utility.} is similar to that of the transaction rate control problem in Tangle. In IoT applications, heterogeneous IoT devices with different computing power are agents, and the transaction rate controller is the principal. Agents want to add new transactions in Tangle at the maximum possible rate. On the other hand, the principal wants to control the rate of new transactions to reduce network congestion and spamming. 
	As the principal cannot observe the computing power of agents, it also has to incentivize agents to reveal their computing power truthfully. A truth-telling mechansim ensures that agents with high computing power solve a difficult PoW. Incentivizing agents to solve a difficult PoW has two benefits: (1) it moderates the rate of new transactions in Tangle (2) it is difficult for an adversary to tamper with transactions in Tangle, making Tangle more secure. We use a simple fairness measure for transaction rate control: it trades off the rate of new transactions with compensation (WoT) given to the agents for solving PoW. To control the transaction rate, the principal assigns a PoW difficulty level to each agent based on their revealed computing power. To ensure the mechanism is truth-telling, the principal compensates agents' PoW using an appropriate WoT.
	
	To summarize, the information asymmetry between the transaction rate controller and  IoT devices motivates PAP with adverse selection as a suitable mechanism for controlling the transaction rate in Tangle. It yields a tractable linear program~\cite{2015:DB} with useful underlying structures that the principal can exploit to speed up the computation.
	
	\subsection*{Related Work}
	Several works study the transaction rate control problem for distributed ledgers.~\cite{2016-DK} formulates a mechanism to control the PoW difficulty level for blockchain that ensures stable average block times. \cite{2017-DM-et-al} proposes a difficulty adjustment algorithm for blockchain to disincentivize the miners from coin-hopping attacks: a malicious miner increases his mining profits while at the same time increasing the average delay between blocks. The transaction rate problem has also been studied from the agent's perspective, e.g.,  agents optimizing their contribution of computing power to the mining process~\cite{2020:EB-et-al}.
	
	For the DAG-based distributed ledger,~\cite{2019-LV-et-al} proposed an adaptive rate control for Tangle. Their scheme increases or decreases the difficulty level of the PoW depending on the historical transaction rate of an agent. \cite{2021:MJ-et-al} uses a utility maximization approach to control the rate of transactions in Tangle using a suitable choice of network performance metric. Their model assumes that the computing power of the agent is known to the transaction rate regulator. \cite{2021:AC-et-al} borrows an idea from wireless networks to control the rate of transactions using an access control scheme. Congestion control has also been studied in the context of wireless communication. For example, \cite{2001:XL-et-al}, \cite{2006:AF-et-al} discuss optimal resource sharing among multiple users for a high quality of service.
	
	To the best of our knowledge, a PAP-based approach to studying the transaction rate control problem in Tangle has not been explored in the literature. The PAP framework allows us to model the information asymmetry between the distributed ledger's users and the transaction rate controller. It yields a tractable mixed-integer optimization problem that the principal can decompose into multiple linear programs. The PAP also allows us to analyze the structure of decision variables; this is beneficial in reducing the search space dimension and decreasing the computation cost. 
 
 Finally, our proposed transaction rate controller is compatible with the IEEE standard 2144.1-2020~\cite{2021:IEEE}. This standard specifies a framework for blockchain-based data management for IoT applications. The IEEE standard  2144.1-2020 identifies the following stakeholders for a typical IoT data collaborative ecosystem: data owners, data consumers, service providers, regulators/policymakers and other stakeholders. Regarding IEEE  2144.1-2020, our proposed transaction rate controller can be considered a regulator/policy maker. The regulator has higher computing power than IoT devices; therefore, it can take the role of transaction rate controller in Tangle. Higher computing power is required to solve our proposed transaction rate control mechanism; using an independent regulator as a transaction rate controller reduces the computational cost of IoT devices.
	
	\subsection*{Organization and Main Results}
	Sec.\ref{sec:problem-statement} describes the Tangle protocol and the PAP approach to solve the transaction rate control problem in Tangle. The PAP is a mixed-integer program: the PoW difficulty level takes values in a finite set, whereas the WoT takes values from the subset of real numbers. We show that for a fixed choice of PoW difficulty level, the PAP for transaction rate control in Tangle is a linear program. This facilitates efficient computation of the optimal solution using standard linear program solvers.
	
	Sec.\ref{sec:structural-results} exploits the structure of the PAP to characterize the decision variables. Our first structural result shows that the optimal PoW difficulty level increases with computing power. The second result shows that the optimal WoT assigned to the agents increases with computing power. The principal can exploit the results to reduce the search space for the transaction rate control problem and decrease the computation cost.
	
	Sec.\ref{sec:simulations} illustrates the application of PAP for the transaction rate control problem in Tangle using numerical examples. We also apply the structural result from Sec.\ref{sec:structural-results} to reduce the dimension of the search space of the decision variables. We incorporate our proposed transaction rate control mechanism in  Tangle and simulate its dynamics. This yields insight into the impact of the transaction rate mechanism on the average approval time of a tip transaction. Finally, we compare the proposed transaction control mechanism with the fixed transaction control scheme in~\cite{2018:SP}, which allocates WoT as a linear function of PoW.
	
	\section{Transaction Rate Control Problem in Tangle}
	\label{sec:problem-statement}
	\begin{figure*}
\centering
\begin{subfigure}[b]{0.2\linewidth}
\begin{center}
\scalebox{0.4}{
\begin{tikzpicture}[auto, thick]
    \foreach \place/\name in {{(-3,0)/0}, {(0,3)/1}, {(0.5,1)/2}, {(0,-1)/3}, {(0.5,-3)/4}}
    \node[transactions, fill=visible] (\name) at \place {\huge \name};
    \draw[decoration={markings,mark=at position 1 with
    {\arrow[scale=2,>=stealth]{>}}},postaction={decorate}] (1) to[out=200,in=70] (0);
    \draw[decoration={markings,mark=at position 1 with
    {\arrow[scale=2,>=stealth]{>}}},postaction={decorate}] (1) to[out=210,in=50] (0);
    
    \draw[decoration={markings,mark=at position 1 with
    {\arrow[scale=2,>=stealth]{>}}},postaction={decorate}] (2) to[out=180,in=30] (0);
    \draw[decoration={markings,mark=at position 1 with
    {\arrow[scale=2,>=stealth]{>}}},postaction={decorate}] (2) to[out=200,in=0] (0);
    
    \draw[decoration={markings,mark=at position 1 with
    {\arrow[scale=2,>=stealth]{>}}},postaction={decorate}] (3) to[out=170,in=330]  (0);
    \draw[decoration={markings,mark=at position 1 with
    {\arrow[scale=2,>=stealth]{>}}},postaction={decorate}] (3) to[out=150,in=350] (0);

    \draw[decoration={markings,mark=at position 1 with
    {\arrow[scale=2,>=stealth]{>}}},postaction={decorate}] (4) to[out=170,in=300] (0);
    \draw[decoration={markings,mark=at position 1 with
    {\arrow[scale=2,>=stealth]{>}}},postaction={decorate}] (4) to (0);
    \foreach \source/\dest in {}
    \draw[decoration={markings,mark=at position 1 with
{\arrow[scale=2,>=stealth]{>}}},postaction={decorate}] (\dest) -- (\source);
    
    \draw[dotted] (-0.5,-4) rectangle (1.3,4);
    \node at (0.4,-5) {\huge $t=1$};
\end{tikzpicture}}
\caption{$G_1(V_1,E_1)$}
\end{center}
\end{subfigure}
\begin{subfigure}[b]{0.3\linewidth}
\centering
\begin{center}
\scalebox{0.4}{
\begin{tikzpicture}[auto, thick]
    \foreach \place/\name in {{(-3,0)/0}, {(0,3)/1}, {(0.5,1)/2}, {(0,-1)/3}, {(0.5,-3)/4}, {(3,4)/5},{(4,2)/6}, {(3,0)/7}, {(3.5,-2.8)/8}}
    \node[transactions, fill=visible] (\name) at \place {\huge \name};
    \draw[decoration={markings,mark=at position 1 with
    {\arrow[scale=2,>=stealth]{>}}},postaction={decorate}] (1) to[out=200,in=75] (0);
    \draw[decoration={markings,mark=at position 1 with
    {\arrow[scale=2,>=stealth]{>}}},postaction={decorate}] (1) to[out=210,in=50] (0);
    
    \draw[decoration={markings,mark=at position 1 with
    {\arrow[scale=2,>=stealth]{>}}},postaction={decorate}] (2) to[out=180,in=30] (0);
    \draw[decoration={markings,mark=at position 1 with
    {\arrow[scale=2,>=stealth]{>}}},postaction={decorate}] (2) to[out=200,in=0] (0);
    
    \draw[decoration={markings,mark=at position 1 with
    {\arrow[scale=2,>=stealth]{>}}},postaction={decorate}] (3) to[out=170,in=330]  (0);
    \draw[decoration={markings,mark=at position 1 with
    {\arrow[scale=2,>=stealth]{>}}},postaction={decorate}] (3) to[out=150,in=350] (0);

    \draw[decoration={markings,mark=at position 1 with
    {\arrow[scale=2,>=stealth]{>}}},postaction={decorate}] (4) to[out=170,in=300] (0);
    \draw[decoration={markings,mark=at position 1 with
    {\arrow[scale=2,>=stealth]{>}}},postaction={decorate}] (4) to (0);
    
    \foreach \source/\dest in {1/5, 1/6, 2/6, 3/7, 3/8, 2/8}
    \draw[decoration={markings,mark=at position 1 with
{\arrow[scale=2,>=stealth]{>}}},postaction={decorate}] (\dest) -- (\source);
    
    \draw[dotted] (-0.5,-4) rectangle (1.3,4);
    \node at (0.4,-5) {\huge $t=1$};
    \draw[dotted] (2.3,-4) rectangle (4.5,5);
    \node at (3.4,-5) {\huge $t=2$};
\end{tikzpicture}
}
\caption{$G_2(V_2,E_2)$}
\end{center}
\end{subfigure}
\begin{subfigure}[b]{0.4\linewidth}
\begin{center}
\scalebox{0.4}{
\begin{tikzpicture}[auto, thick]
    \foreach \place/\name in {{(-3,0)/0}, {(0,3)/1}, {(0.5,1)/2}, {(0,-1)/3}, {(0.5,-3)/4}, {(3,4)/5},{(4,2)/6}, {(3,0)/7}, {(3.5,-2.8)/8}, {(6.2,2.3)/9}, {(7,1)/10}, {(6.5,-1)/11}, {(6.7,-3.5)/12}}
    \node[transactions, fill=visible] (\name) at \place {\huge \name};
    \draw[decoration={markings,mark=at position 1 with
    {\arrow[scale=2,>=stealth]{>}}},postaction={decorate}] (1) to[out=200,in=75] (0);
    \draw[decoration={markings,mark=at position 1 with
    {\arrow[scale=2,>=stealth]{>}}},postaction={decorate}] (1) to[out=210,in=50] (0);
    
    \draw[decoration={markings,mark=at position 1 with
    {\arrow[scale=2,>=stealth]{>}}},postaction={decorate}] (2) to[out=180,in=30] (0);
    \draw[decoration={markings,mark=at position 1 with
    {\arrow[scale=2,>=stealth]{>}}},postaction={decorate}] (2) to[out=200,in=0] (0);
    
    \draw[decoration={markings,mark=at position 1 with
    {\arrow[scale=2,>=stealth]{>}}},postaction={decorate}] (3) to[out=170,in=330]  (0);
    \draw[decoration={markings,mark=at position 1 with
    {\arrow[scale=2,>=stealth]{>}}},postaction={decorate}] (3) to[out=150,in=350] (0);

    \draw[decoration={markings,mark=at position 1 with
    {\arrow[scale=2,>=stealth]{>}}},postaction={decorate}] (4) to[out=170,in=300] (0);
    \draw[decoration={markings,mark=at position 1 with
    {\arrow[scale=2,>=stealth]{>}}},postaction={decorate}] (4) to (0);
    \foreach \source/\dest in {1/5, 1/6, 2/6, 3/7, 3/8, 2/8, 7/9, 8/9, 6/10, 7/10, 4/11, 8/11, 7/12, 8/12}
    \draw[decoration={markings,mark=at position 1 with
{\arrow[scale=2,>=stealth]{>}}},postaction={decorate}] (\dest) -- (\source);
    
    \draw[dotted] (-0.5,-4) rectangle (1.3,4);
    \node at (0.4,-5) {\huge $t=1$};
    \draw[dotted] (2.3,-4) rectangle (4.5,5);
    \node at (3.4,-5) {\huge $t=2$};
    \draw[dotted] (5.5,-4) rectangle (7.7,3.5);
    \node at (6.6,-5) {\huge $t=3$};
\end{tikzpicture}
}
\caption{$G_3(V_3,E_3)$}
\end{center}
\end{subfigure}

\begin{subfigure}[b]{0.4\linewidth}
\begin{center}
\scalebox{0.4}{
\begin{tikzpicture}[auto, thick]
    \foreach \place/\name in {{(-3,0)/0}, {(0,3)/1}, {(0.5,1)/2}, {(0,-1)/3}, {(0.5,-3)/4}, {(3,4)/5},{(4,2)/6}, {(3,0)/7}, {(3.5,-2.8)/8}, {(6.2,2.3)/9}, {(7,1)/10}, {(6.5,-1)/11}, {(6.7,-3.5)/12}, {(9.5,4)/13}, {(10,2)/14}, {(9.5,-1)/15}, {(9.7,-3)/16}}
    \node[transactions, fill=visible] (\name) at \place {\huge \name};
    \draw[decoration={markings,mark=at position 1 with
    {\arrow[scale=2,>=stealth]{>}}},postaction={decorate}] (1) to[out=200,in=75] (0);
    \draw[decoration={markings,mark=at position 1 with
    {\arrow[scale=2,>=stealth]{>}}},postaction={decorate}] (1) to[out=210,in=50] (0);
    
    \draw[decoration={markings,mark=at position 1 with
    {\arrow[scale=2,>=stealth]{>}}},postaction={decorate}] (2) to[out=180,in=30] (0);
    \draw[decoration={markings,mark=at position 1 with
    {\arrow[scale=2,>=stealth]{>}}},postaction={decorate}] (2) to[out=200,in=0] (0);
    
    \draw[decoration={markings,mark=at position 1 with
    {\arrow[scale=2,>=stealth]{>}}},postaction={decorate}] (3) to[out=170,in=330]  (0);
    \draw[decoration={markings,mark=at position 1 with
    {\arrow[scale=2,>=stealth]{>}}},postaction={decorate}] (3) to[out=150,in=350] (0);

    \draw[decoration={markings,mark=at position 1 with
    {\arrow[scale=2,>=stealth]{>}}},postaction={decorate}] (4) to[out=170,in=300] (0);
    \draw[decoration={markings,mark=at position 1 with
    {\arrow[scale=2,>=stealth]{>}}},postaction={decorate}] (4) to (0);
    \foreach \source/\dest in {1/5, 1/6, 2/6, 3/7, 3/8, 2/8, 7/9, 8/9, 6/10, 7/10, 4/11, 8/11, 7/12, 8/12, 5/13, 10/13, 5/14, 10/14, 11/16, 12/16}
    \draw[decoration={markings,mark=at position 1 with
{\arrow[scale=2,>=stealth]{>}}},postaction={decorate}] (\dest) -- (\source);
    \draw[decoration={markings,mark=at position 1 with
{\arrow[scale=2,>=stealth]{>}}},postaction={decorate}] (15) to[out=170,in=290] (10);
    \draw[decoration={markings,mark=at position 1 with
{\arrow[scale=2,>=stealth]{>}}},postaction={decorate}] (15) to[out=120,in=340] (10);

    \draw[dotted] (-0.5,-4) rectangle (1.3,4);
    \node at (0.4,-5) {\huge $t=1$};
    \draw[dotted] (2.3,-4) rectangle (4.5,5);
    \node at (3.4,-5) {\huge $t=2$};
    \draw[dotted] (5.5,-4) rectangle (7.7,3.5);
    \node at (6.6,-5) {\huge $t=3$};
    \draw[dotted] (8.5,-4) rectangle (10.7,5);
    \node at (9.6,-5) {\huge $t=4$};

\end{tikzpicture}
}
\caption{$G_4(V_4,E_4)$}
\end{center}
\end{subfigure}
\caption{Protocol for tip transaction selection and approval in Tangle. Consider a dynamic DAG $G_t(V_t,E_t)$ representing Tangle at time $t$. Each node of $G_t(V_t,E_t)$ represents a transaction. Node 0 is the genesis node at $t=0$. At each time $t$, new transactions join the ledger by randomly selecting two tip transactions and approving them. This corresponds to formation of a directed edge. (a) Node $1,2,3,4$ join the ledger at $t=1$. Tip node for $G_0=(V_0,E_0)$ is $0$. (b) Node $5,6,7,8$ join the  at $t=2$. Tip nodes for $G_1=(V_1,E_1)$ are $1,2,3,4$. (c) Node $9,10,11,12$ join the ledger at $t=3$. Tip nodes for $G_2=(V_2,E_2)$ are $4,5,6,7,8$. (d) Node $13,14,15,16$ join the ledger at $t=4$. Tip nodes for $G_3=(V_3,E_3)$ are $5,9,10,11,12$. } 
\label{fig:tangle-evolution}
\end{figure*}
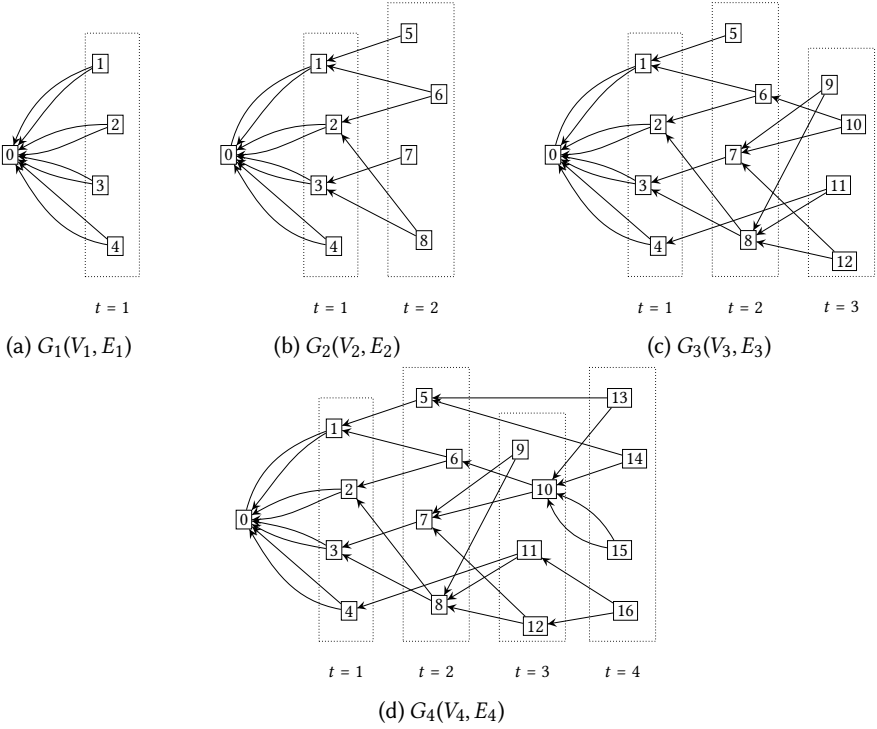
	
	This section describes the PAP formulation for controlling the rate of new transactions in Tangle. Sec.\ref{sec:tangle-protocol} describes the Tangle protocol, and Sec.\ref{sec:pap} describes the PAP formulation for the transaction rate control problem in Tangle.
	\subsection{Tangle protocol}
	\label{sec:tangle-protocol}
	
	Tangle can be abstracted as a time-evolving DAG. Fig.~\ref{fig:tangle-evolution} shows an example of the time evolution of  Tangle. 
	
	Consider a time evolving DAG~$G_t(V_t,E_t)$ representing  Tangle at time~$t$
	where $t\in\mathbb{Z}^+$ denotes discrete time.
	Each node in the graph corresponds to a transaction/record stored on Tangle. At $t=0$, $G_0(V_0,E_0)$ denotes the \textit{genesis} graph. We assume $|V_0|=1, |E_0|=0$. At each discrete time instant  $t$, new transactions are added in Tangle by agents (IoT devices). To join Tangle, each new transaction at time $t$ chooses two \textit{tip transactions} from $G_{t-1}(V_{t-1},E_{t-1})$ randomly, and \textit{approves} them. The two tip nodes are selected independently with repetition. Here, a tip transaction means a transaction with no incoming edges. Approving a tip transaction in Tangle means verifying that the transaction is valid, i.e., there is no double-spending. If the two randomly chosen transactions are valid, the new transaction forms a directed edge to each randomly chosen tip transaction. We will assume that the chosen tip transactions are always valid, as handling double-spending attacks are beyond the scope of this study. To summarize, every new transaction forms outgoing edges from itself to two randomly chosen tip transactions to join Tangle. 
	
	When adding a new transaction, agents must show PoW by solving a hash puzzle. This is to prevent spamming\footnote{Solving the  hash puzzle for PoW takes a finite time. Hence, an agent's rate of new transactions is restricted through PoW.}. Has puzzles also ensure that tampering\footnote{To tamper with a transaction, a malicious entity would need to re-solve the hash puzzle for that transaction and also for all future transactions that approve it. This is because a solution of the hash puzzle of a transaction depends on the hash of the previous transactions} a transaction in Tangle is difficult~\cite{2017:FH-et-al}. Based on the PoW difficulty level, the transaction is assigned a weight. During random selection of the two tip transactions, the probability of choosing a particular tip transaction is directly proportional to its weight. Hence, agents optimally choose the PoW difficulty level based on their preference for WoT. In the current implementation of Tangle~\cite{2018:SP}, the relation between the WoT and its PoW difficulty level is a fixed linear function. We propose a modified mechanism to assign the PoW difficulty level and corresponding WoT based on the computing power of an agent. The mechanism is also truth-telling, i.e., agents truthfully reveal their computing power to the principal.
 
	\textit{Remarks: }Approving tip transactions takes finite time and leads to delay, i.e., transaction entering Tangle at time $t$ is available as a tip transaction at time $t+h,\;h\geq 1$. In this paper, we assume a constant delay $h=1$ for simplicity. The choice of $h$ does not affect the transaction rate control problem. It only affects the dynamics of Tangle.
	
	\subsection{PAP approach to the transaction rate control problem in Tangle}
	\label{sec:pap}
	Our PAP formulation of the transaction rate control problem in Tangle is constructed as follows:
	\begin{align}
	\text{PAP}&=\begin{cases}\text{Principal} &=\text{Transaction rate controller}\\
	\text{Agents}&=\text{Users (devices) adding transactions in Tangle}
	\end{cases}
	\end{align}
        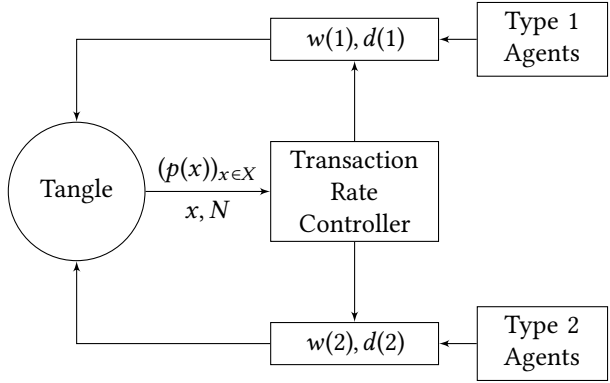
\begin{figure}
		\begin{center}
			\begin{tikzpicture}
			\tikzstyle{arrow} = [draw, -latex']
			\tikzstyle{user} = [draw, rectangle, text width=1.5cm, align=center];
			\tikzstyle{pow} = [draw, rectangle, text width=2cm, align=center]
			\node (ledger) at (-1.2,0) [draw, circle, text width=1.5cm, align=center] {Tangle};
			\node (controller) at (2.5,0) [draw, rectangle, text width = 2cm, align=center] {Transaction Rate Controller};
			\draw[arrow] (ledger) -- node[above] {$(p(\compPower))_{\compPower\in\compPowerSet}$} node[below]{$\compPower,\totalAgents$} (controller);
			\foreach \x in {1,2}
			{
				\node (pow\x) at (2.5,6-4*\x) [pow] {$\wot(\state_\x),\pow(\state_\x)$};
				\node (user\x) at (5,6-4*\x) [user] {Type $\state_\x$ Agents};
				\draw [arrow] (user\x) -- (pow\x);
				\draw [arrow] (pow\x) -|(ledger);
				\draw [arrow] (controller) -- (pow\x);
			}
			\end{tikzpicture}
		\end{center}
		\caption{Block diagram for our proposed transaction rate control mechanism in Tangle. For simplicity, the block diagram illustrates the case when the different types of agents $\compPowerSet$ are equal to $2$. Here, agents' type refers to their computing power. The transaction rate controller block receives the fraction of different types of agents $(p(\compPower))_{\compPower\in\compPowerSet},\compPower$, total agents $\totalAgents$ from Tangle as an input. The controller then assigns a difficulty level of PoW $\pow(\compPower)$ to each agent depending on its type $\compPower\in\compPowerSet$. To compensate the PoW, new transactions created by the agent are given a weight $\wot(\compPower)$. New transactions are part of the set of tip transactions. Having a higher weight increases the chance of a tip transaction being selected for approval by other transactions.}
		\label{fig:block_diagram}
	\end{figure}
	Consider a time-evolving Tangle where multiple agents (IoT devices) add new transactions into the distributed ledger as shown in Fig.~\ref{fig:block_diagram}. Agents are heterogeneous, i.e., agents have different computing power.
	
	Let $\compPower\in \compPowerSet:=\{1,2,\ldots,\stateCardinality\}$ denote the computing power of an agent. Let $\totalAgents$ denote the number of agents, and let $\prob(\compPower)$ denote the fraction of agents having computing power $\compPower$. We define the probability vector $\prob\in\R^\stateCardinality$ as:
    \begin{align}
        \label{eq:prob-vector}
        \prob = [\prob(1), \prob(2), \ldots, \prob(\stateCardinality)]
    \end{align}
    To add a new transaction in Tangle, agents have to satisfy the PoW requirement by solving a hash puzzle: search for a nonce~\cite{2008:SN} that results in hash code starting with a desired number of zeros. The more difficult the puzzle, the longer it takes for an agent to solve it. Hence, to control the rate of new transactions in Tangle, the principal (transaction rate controller) adjusts the PoW difficulty level for each agent. Let $\pow(\compPower)\in\powSet:=\{1,2,\ldots,\powCardinality\}$ denotes the PoW difficulty level assigned to an agent with computing power $\compPower$. We define the PoW difficulty vector $\pow\in\powSet^{\stateCardinality}$ as:
	\begin{align}
	\label{eq:effort-vector}
	\pow = [\pow(1), \pow(2), \ldots, \pow(\stateCardinality)]
	\end{align}
	The PoW difficulty level $\pow(\compPower)$ corresponds to the number of zeros required at the beginning of the hash code. Higher difficulty decreases the rate at which an agent can add new transactions to the distributed ledger. To compensate for the agent's PoW, the principal assigns a WoT $\wot(\compPower)\in\wotSet:=[1,\infty)$ to the transaction added by an agent with computing power~$\compPower$. We define the WoT vector $\wot\in\wotSet^{\stateCardinality}$ as:
    \begin{align}
        \label{eq:wot-vector}
        \wot = [\wot(1), \wot(2),\ldots, \wot(\stateCardinality)]
    \end{align}
    Every new transaction must approve two existing tip transactions (transactions with no incoming edges) to join Tangle. In the random tip selection strategy~\cite{2018:SP}, the two tip transactions for approval are chosen randomly with a probability proportional to the weight of a tip transaction. Hence, the higher the weight of a tip transaction, the faster it gets selected for approval by others on average.
	
	\subsection*{Utility function for IoT devices (agents)} 
	We consider the following general structure of a  utility function for an agent with computing power~$\compPower$:
	\begin{align}
	\label{eq:agent-utility-general}
	\utilityAgent(\wot(\compPower),\pow(\compPower),\compPower) := \wot(\compPower)-g\bracketRound{\pow(\compPower),\compPower}
	\end{align}
	Here, $\pow, \wot$ are defined in \eqref{eq:effort-vector} and \eqref{eq:wot-vector}, respectively. 
 If the WoT $\wot(\compPower)$ is large, the tip transaction gets quickly approved by new transactions. Hence, agents prefer a higher WoT $\wot(\compPower)$. Here, the increasing function $g(\cdot)$ models the associated cost for an agent to satisfy the PoW requirement. If the PoW difficulty level is high, the agent has to spend large computing power to satisfy the PoW requirement. Hence, agents prefer low effort $\pow(\compPower)$. Moreover, we assume that $g_\pow(\cdot)$ is decreasing in $\compPower$. This models that the marginal cost due to increase in PoW difficulty level decreases with computing power $\compPower$. 
 
	
	\subsection*{Formulation of a PAP for the transaction rate control in Tangle} In our formulation, we assume that the computing power of agents is known to the principal. Still, the computing power of agents is unobserved\footnote{The computing power of an agent is typically estimated by the number of transactions added by an agent in the recent history~\cite{2020:SP-et-al}. It may not be possible to estimate the computing power for IoT applications. This is because IoT devices are not dedicated miners. Hence, they need not add new transactions at the maximum possible rate; IoT devices could be switching between multiple tasks or could be in standby mode. 
		Our approach is for the transaction rate controller (principal) to incentivize agents to report their computing power truthfully. Truth-telling is ensured through incentive constraints~\eqref{eq:pap-general-incentive-constraint} of the PAP.} by the principal (transaction rate controller). Therefore, the principal has to design a truth-telling mechanism that maximizes its utility. To control the rate of new transactions, the principal solves the PAP~\eqref{eq:pap-general} to assign PoW difficulty levels~$\pow(\compPower)$ to different agents, and WoT~$\wot(\compPower)$ for the transactions added by them into Tangle. The incentive constraints ensure that agents truthfully choose the PoW assigned for their computing power~$\compPower$.
	
	The PAP for transaction rate control is the following constrained optimization problem:

        \begin{subequations}
		\label{eq:pap-general}
		\begin{align}
		\label{eq:pap-general-objective}
		&\min_{\substack{\pow,\wot,\\\forall \compPower\in \compPowerSet}} \fairness(\compPower,\pow,\wot,\prob,\totalAgents)\\
		\label{eq:pap-general-incentive-constraint}
		\text{s.t.}\quad &\compPower=\arg\max_{{\bar{\compPower}\in \compPowerSet}} \utilityAgent(\wot(\compPower),\pow(\compPower),\compPower),\,\forall \compPower\in \compPowerSet\\
		\label{eq:pap-general-participation-constraint}
		&\utilityAgent(\wot(\compPower),\pow(\compPower),\compPower)\geq u_0,\forall \compPower\in \compPowerSet
		\end{align}
	\end{subequations}
        Here, $\compPower$ is the computing power of an agent. $\prob,\pow,\wot$ are defined in \eqref{eq:prob-vector}, \eqref{eq:effort-vector} and \eqref{eq:wot-vector}, respectively. $\totalAgents$ is the total number of agents. $\utilityAgent_0$ is the base utility level of agents; if the utility is above $\utilityAgent_0$, agents are willing to participate in the distributed ledger; otherwise, agents would opt out, i.e., they would not use Tangle to store their transactions. $\fairness(\cdot)$ is a fairness function used by the principal for transaction rate control. $\utilityAgent(\cdot)$ is defined in \eqref{eq:agent-utility-general}. \eqref{eq:pap-general-incentive-constraint} is known as incentive constraint for the transaction rate control problem; it ensures that agents are truthful. \eqref{eq:pap-general-participation-constraint} is known as participation constraint; it guarantees a base utility level for all agents. 

        \subsection*{Structure of transaction rate control problem in Tangle}
        The PAP~\eqref{eq:pap-general} is a mixed-integer optimization problem. It has useful structure: the optimal solution of the PAP~\eqref{eq:pap-general} is increasing in computing power $\compPower$; we would discuss these results in Sec.\ref{sec:structural-results}. The structural result reduces the computation cost for solving the transaction rate control problem. Moreover, we derive the structural results only using the incentive constraints~\eqref{eq:pap-general-incentive-constraint}, which are a function of the utility of agents~\eqref{eq:agent-utility-general}. This allows the principal to choose different fairness measures $f$ in the PAP~\eqref{eq:pap-general} for controlling the transaction rate in Tangle. Thus the results in this paper apply to a large class of fairness measures for the transaction rate control problem in Tangle.

    \subsection*{Implementation of the transaction rate control problem in Tangle}
    The PAP~\eqref{eq:pap-general} is a mixed-integer optimization problem. The principal first solves the PAP~\eqref{eq:pap-general} to obtain the optimal WoT $\wot^*(\compPower)$ for each possible choice of effort $\pow\in\powSet^{\stateCardinality}$ (defined in~\eqref{eq:effort-vector}). For a fixed effort $\pow$, the PAP is a continuous optimization problem; it can be solved efficiently using standard solvers. Our model is suitable for consortium distributed ledgers \cite{2018:OD-et-al}: a hybrid of public and private distributed ledgers. This is because it can be computationally intensive for IoT devices to solve the PAP~\eqref{eq:pap-general}. In a consortium distributed ledger, a single private agent can be the transaction rate controller (principal). Here, the private agent could be the secure facilitator of the distributed ledger with higher computing power\footnote{Our proposed transaction rate controller is compatible with the IEEE standard 2144.1-2020~\cite{2021:IEEE}: a framework for blockchain-based data management for IoT applications. The IEEE standard  2144.1-2020 identifies the following stakeholders for a typical IoT data collaborative ecosystem: data owners, data consumers, service providers, regulators/policymakers and other stakeholders. According to IEEE standard 2144.1-2020, our proposed transaction rate controller can be viewed as a regulator/policy maker with significantly higher computing power. IoT devices with limited computing power can be categorized as data owners who add transactions in Tangle.}. Moreover, the transaction rate controller can use structural results (discussed in Sec.\ref{sec:structural-results}) to solve the PAP~\eqref{eq:pap-general} approximately.
	
	
	
	\subsection*{Truth-telling and implications for Tangle} The incentive constraint~\eqref{eq:pap-general-incentive-constraint} ensures that the mechanism is truth-telling\footnote{The proposed transaction rate control mechanism does not prevent agents from colluding, i.e., it may be advantageous for two agents to combine their computing power and act as a single agent. Preventing malicious collusion of agents is a subject of future work.}. This is achieved by maximizing the utility of each agent when they perform PoW assigned for their computing power. If~\eqref{eq:pap-general-incentive-constraint} is omitted, agents can increase their utility by choosing a PoW difficulty level assigned for different computing power. In such a case, the actual transaction rate can exceed the optimal transaction rate solved by the principal. This can lead to network congestion and delay broadcasting new transactions among all agents. 
	Therefore, \eqref{eq:pap-general-incentive-constraint} ensures that an agent with computing power $\compPower$ will be worse off in terms of its preference for PoW and WoT if it doesn't tell truth, i.e., $\utilityAgent(\pow(\compPower),\wot(\compPower),\compPower)$ is better than $\utilityAgent(\pow(\bar{\compPower}),\wot(\bar{\compPower}),\compPower)$. 

	
	\textit{Summary}.	This section formulated the transaction rate control problem in Tangle as the PAP~\eqref{eq:pap-general}; it is a mixed-integer optimization problem.
	In Sec.\ref{sec:structural-results}, we will exploit the structure of the PAP to formulate a mixed-integer optimization problem of smaller dimensions.
	
	
	\section{Structural Analysis  of 
		the Transaction Rate Control Problem}
    \label{sec:structural-results}
	In the previous section, we formulated the PAP~\eqref{eq:pap-general} for the transaction rate control problem in Tangle. The PAP~\eqref{eq:pap-general} is a mixed-integer optimization problem; the principal can obtain the optimal solution by solving $\powCardinality^\stateCardinality$ continuous optimization problems for each possible values of PoW difficulty vector $\pow$ (defined in \eqref{eq:effort-vector}). Moreover, the search space for each continuous optimization problem is in $\R^\stateCardinality$. We present two structural results on the optimal solution of the PAP~\eqref{eq:pap-general}. These structural results guarantee that the optimal WoT $\wot^*(\compPower)$ and the optimal PoW difficulty level $\pow^*(\compPower)$ are increasing in computing power $\compPower$. The structural results can be used to decrease the computation cost of solving the PAP~\eqref{eq:pap-general}.
	
	Our first result deals with the structure of the optimal PoW difficulty level $\pow^*(\compPower)$ assigned to an agent with computing power $\compPower$. We show that the optimal PoW difficulty level $\pow^*(\compPower)$ is non-decreasing in $\compPower$. Moreover, we derive the structural result only using the incentive constraints~\eqref{eq:pap-general-incentive-constraint}, which are a function of the utility of agents~\eqref{eq:agent-utility-general}. This allows the principal to choose different fairness measures $f$ in the PAP~\eqref{eq:pap-general} for controlling the transaction rate in Tangle.
	\begin{theorem}
		\label{thm:effort-monotone}
		The optimal PoW difficulty level $\pow^*(\compPower)$ assigned to an agent with computing power $\compPower$ by the PAP~\eqref{eq:pap-general} to control transaction rate is increasing in $\compPower$. For the PAP~\eqref{eq:pap-general}, this reduces the search-space for PoW difficulty vector $\pow^*$ (defined in \eqref{eq:effort-vector}) from $\powCardinality^\stateCardinality$ to $\sum_{i=1}^m \binom{n}{i}\binom{m}{i}$. Moreover, this result holds for any fairness measure $f$ in the objective~\eqref{eq:pap-general-objective}.
	\end{theorem}
        \begin{proof}
			\begin{align*}
			&\text{Consider computing power $\compPower, \bar{\compPower}\in\compPowerSet$ s.t. $\compPower<\bar{\compPower}$. Constraints~\eqref{eq:pap-specific-incentive-constraint} imply}\\
		&\wot({\compPower})-g(\pow({\compPower}),\compPower)\geq \wot(\bar{\compPower})-g(\pow(\bar{\compPower}),\compPower) ,\;\text{ and }
			-\bracketRound{\wot({\compPower})-g(\pow({\compPower}),\bar{\compPower}})\geq -\bracketRound{\wot(\bar{\compPower})-g(\pow(\bar{\compPower}),\bar{\compPower}})\\
            &\text{Adding above two inequalities yields}\;\;
			g(\pow({\compPower}),\compPower)-g(\pow(\bar{\compPower}),\compPower) \leq g(\pow({\compPower}),\bar{\compPower})-g(\pow(\bar{\compPower}),\bar{\compPower})\\
		&\text{As $g_\pow(\cdot)$ is deceasing in $\compPower$, last inequality implies}\; \pow(\compPower) \leq \pow(\bar{\compPower})
			\end{align*}
	\end{proof}	
	The formulation in Sec.\ref{sec:problem-statement} required the principal (transaction rate controller) to solve $\powCardinality^{\stateCardinality}$ continuous optimization problems corresponding to each possible value of PoW difficulty vector $\pow\in\powSet^{\stateCardinality}$ (defined in \eqref{eq:effort-vector}). By using the structural result in Theorem~\ref{thm:effort-monotone}, we can significantly reduce the number of continuous optimization problems to be solved by the principal. Specifically, Theorem~\ref{thm:effort-monotone} ensures that the principal has to solve continuous optimization problems obtained by fixing PoW difficulty vector $\pow$~\eqref{eq:effort-vector} in the PAP~\eqref{eq:pap-general} that satisfy $\pow(i)\leq\pow(j),\forall i<j$. This implies that the number of continuous optimization problems to be solved by the transaction rate controller is reduced from $m^n$ to $\sum_{i=1}^m \binom{n}{i}\binom{m}{i}$.
	
	Our second result deals with the structure of the WoT $\wot^*(\compPower)$ assigned to an agent with computing power $\compPower$. We show that $\wot^*(\compPower)$ is non-decreasing in $\compPower$. Moreover, we derive the structural result only using the incentive constraints~\eqref{eq:pap-general-incentive-constraint}, which are a function of the utility of agents~\eqref{eq:agent-utility-general}. This allows the principal to choose different fairness measures $f$ in the PAP~\eqref{eq:pap-general} for controlling the transaction rate in Tangle. 
	\begin{theorem}
		\label{thm:weight-monotone}
		The optimal WoT $\wot^*(\compPower)$ assigned to an agent with computing power $\compPower$ by the PAP~\eqref{eq:pap-general} to control transaction rate is increasing in $\compPower$. Hence, for the PAP~\eqref{eq:pap-general}, we can obtain a constrained optimal WoT vector $\wot^*$ (defined in~\eqref{eq:wot-vector}) within a class of increasing functions at a reduced computation cost. Moreover, this result holds for any fairness measure $f$ in the objective~\eqref{eq:pap-general-objective}.
	\end{theorem}
        \begin{proof}
			\begin{align*}
			&\text{Consider computing power $\compPower, \bar{\compPower}\in\compPowerSet$ s.t. $\compPower<\bar{\compPower}$. Constraints~\eqref{eq:pap-specific-incentive-constraint} imply}\\
			&\wot({\compPower})-g(\pow({\compPower}),\compPower)\geq \wot(\bar{\compPower})-g(\pow(\bar{\compPower}),\compPower)\\
			&\text{Using Theorem~\ref{thm:effort-monotone} and $g(\cdot)$ is increasing in effort $\pow$ imply}\\
			&\wot(\bar{\compPower})-\wot(\compPower)\geq g(\pow(\bar{\compPower}),\bar{\compPower})-g(\pow({\compPower}),\bar{\compPower})\geq 0\Rightarrow \wot(\compPower) \leq \wot(\bar{\compPower})
			\end{align*}
	\end{proof} 
	Theorem~\ref{thm:weight-monotone} can be used to parametrize the WoT $\wot(\compPower)$ within the class of increasing functions in the PAP~\eqref{eq:pap-general}. For example, we can solve for an optimal $\wot^*(\compPower)$ within the class of increasing affine functions of $\compPower$. This decreases the dimension of the search space from $\R^{\stateCardinality}$ to $\R^2$ and provides a constrained optimal solution (constrained to affine increasing functions)  at a reduced computation cost. 
	
	
	To summarize, we presented two structural results on the solution of the transaction rate control problem~\eqref{eq:pap-general} in Tangle. The first result guarantees the monotonicity of the optimal PoW difficulty level $\pow^*(\compPower)$ in computing power $\compPower$. It helps reduce the number of linear programs the principal has to solve for the transaction rate control problem~\eqref{eq:pap-general} in Tangle. The second result guarantees the monotonicity of the optimal WoT $\wot^*(\compPower)$ in computing power $\compPower$. This facilitates the principal to parametrize $\wot(\compPower)$ within a class of increasing functions and obtain a constrained optimal solution at a low computational cost. 
	\section{Numerical Results. Transaction Rate control in Tangle}
	\label{sec:simulations}
	This section illustrates, via numerical examples, our proposed
	transaction rate control mechanism for Tangle.
	We first specify the model parameters and solve the transaction rate control problem~\eqref{eq:pap-general}. Later, we utilize the solution of the PAP~\eqref{eq:pap-general} to simulate the dynamics of an actual Tangle and study the impact of our proposed transaction rate control mechanism on the average approval time of tip transactions. The main takeaway from the simulations are: 1) agents with lower computing power are assigned lower PoW difficulty levels at the cost of a larger average approval time of their transactions, 2) the fixed transaction control in~\cite{2018:SP} where WoT is a linear function of PoW cannot incentivize agents with higher computing power to do a difficult PoW. To incentivize a difficult PoW, the marginal increase in WoT should increase in PoW difficulty level. 

    \subsection*{Model for the transaction rate control problem in Tangle}
	For a demonstration of our model, we work with a specific choice of utility function~\eqref{eq:agent-utility-general} for agents:
	\begin{align}
	\label{eq:agent-utility-specific}	\utilityAgent(\wot(\compPower),\pow(\compPower),\compPower) := \beta\wot({\compPower})-\frac{\exp(\pow({\compPower}))}{\compPower}
	\end{align}
	Here, $\beta\in\R^+$ is a parameter that tunes agents' preference between WoT and PoW. The first term $\beta\wot(\compPower)$ models agents' preference for higher WoT, and the second term $\frac{\exp(\pow({\compPower}))}{\compPower}$ models preference for low PoW $\pow(\compPower)$ by agents. First term $\beta\wot(\compPower)$ is linear in WoT as the probability of selection of tip nodes for approval is directly proportional to WoT (see Tangle protocol in Sec.\ref{sec:tangle-protocol}). Rate of new transaction $\frac{\compPower}{\exp{\pow(\compPower)}}$ is an increasing function of the computing power $\compPower$; it decreases as exponential of the PoW difficulty level $\pow(\compPower)$. This is because PoW difficulty $\pow(\compPower)$ corresponds to a search for hash code starting with $\pow(\compPower)$ number of zeros. Assuming each hash code is equally probable, the probability of finding a hash code starting with $\pow(\compPower)$ number of zeros decreases exponentially with $\pow(\compPower)$. Hence, the cost is an exponential function of $\pow(\compPower)$.

    For simulations, we use the following PAP for the transaction control problem in Tangle with utility function for agents defined in~\eqref{eq:agent-utility-specific}:
    \begin{subequations}
		\label{eq:pap-specific}
		\begin{align}
		\label{eq:pap-specific-objective}
		&\min_{\substack{\pow(\compPower)\in\powSet,\\\wot(\compPower)\in\wotSet,\\\forall \compPower\in \compPowerSet,\\ \wot(1)=1}} \sum_{\compPower\in \compPowerSet}\prob(\compPower)\bracketSquare{\totalAgents\,\rate{\compPower}{\compPower} +\alpha \wot(\compPower))}\\
		\label{eq:pap-specific-incentive-constraint}
		\text{s.t.}\quad &\compPower=\arg\max_{{\bar{\compPower}\in \compPowerSet}} \beta\wot(\bar{\compPower})-\frac{1}{\rate{\compPower}{\bar{\compPower}}},\,\forall \compPower\in \compPowerSet\\
		\label{eq:pap-specific-participation-constraint}
		&\beta\wot(\compPower)-\frac{1}{\rate{\compPower}{\bar{\compPower}}}\geq u_0,\forall \compPower\in \compPowerSet\\
		\intertext{where,}
		&\rate{\compPower}{\bar{\compPower}}:=\frac{\compPower}{\exp(\pow(\bar{\compPower}))}
		\end{align}
	\end{subequations}
	        
	Here, the non-negative parameter  $\alpha\in\R^+$  tunes the fairness measure in objective~\eqref{eq:pap-specific-objective}
    The objective function~\eqref{eq:pap-specific-objective} ensures a trade-off between the rate of new transactions and WoT assigned to different agents. $\rate{\compPower}{\bar{\compPower}}$ denotes the rate at which an agent with computing power $\compPower$ can add new transactions if it claims its computing power to be $\bar{\compPower}$. The first term $\totalAgents\rate{\compPower}{\compPower}$ controls the rate of new transactions by adjusting the PoW difficulty level for each agent. This assumes that agents are truthful; truth-telling is ensured through incentive constraint~\eqref{eq:pap-specific-incentive-constraint} (discussed later). The rate of new transactions is directly proportional to the number of agents $\totalAgents$ and the computing power $\compPower$ available to them. It is inversely proportional to the exponential of PoW difficulty level $\pow(\compPower)$. This is because PoW difficulty $\pow(\compPower)$ corresponds to a search for hash code starting with $\pow(\compPower)$ number of zeros. Assuming each hash code is equally probable to occur, the probability of finding a hash code starting with $\pow(\compPower)$ number of zeros is proportional to $\frac{\compPower}{\pow(\compPower)}$. Therefore, minimizing the first term assigns higher PoW to agents and controls transaction rate. Minimizing the second term $\alpha \wot(\compPower)$ ensures that agents are assigned similar WoT. This ensures fairness amongst agents with different computing power. This is because the probability of a tip transaction being selected for approval by a new transaction is proportional to its WoT (see Tangle protocol in Sec.\ref{sec:tangle-protocol}). Hence,  $\wot(\compPower)$ affects the average time before a tip transaction gets approved. As the importance of WoT is only relative, we normalize the WoT with respect to the WoT of the agent with the smallest computing power. This is achieved by adding a constraint $\wot(1)=1$. The participation constraint~\eqref{eq:pap-specific-participation-constraint} guarantees a base utility level for all agents; otherwise, agents would opt out of the distributed ledger, i.e., they would not use Tangle to store their transactions.

    \subsection*{Simulation setup for the transaction rate control in Tangle}
    The PAP~\eqref{eq:pap-general} is a mixed-integer optimization problem. 
    For a fixed effort $\pow$ (defined in~\eqref{eq:effort-vector}), the PAP~\eqref{eq:pap-specific} is a linear optimization program and can be solved efficiently using standard solvers.
 The principal chooses the optimal PoW difficulty level $\pow^*\in\powSet^{\stateCardinality}$ that maximizes its utility. We implement the PAP~\eqref{eq:pap-specific} in MATLAB and use the inbuilt optimization toolbox to obtain the optimal WoT vector $\wot$ (defined in \eqref{eq:wot-vector}) and the optimal PoW vector $\pow$ (defined in \eqref{eq:effort-vector}). We also simulate the dynamics of Tangle~\cite{2018:SP} (see Sec.\ref{sec:tangle-protocol} for the Tangle protocol) in MATLAB and compute the average approval time of transactions.
	
 We begin with a simulation of the PAP~\eqref{eq:pap-specific} to compute the optimal PoW difficulty level and WoT for agents. The model parameters are listed in Table~\ref{tab:model-parameters}. We solve the PAP~\eqref{eq:pap-specific} for four different values of the number of agents $\totalAgents$. As the PoW difficulty level exponentially affects the objective function~\eqref{eq:pap-specific-objective}, we simulate for $\totalAgents=\{100,1000,10000,100000\}$ to observe a noticeable change in PoW difficulty level. The chosen number of agents $\totalAgents$ is large enough to model a realistic scenario. We use Theorem~\ref{thm:effort-monotone} to reduce the search space for the PoW vector $\pow$ (defined in~\eqref{eq:effort-vector}). Specifically, we solve the PAP~\eqref{eq:pap-specific} for those $\pow$ that satisfy $\pow(i)\leq\pow(j),\forall i<j$. We do not use Theorem~\ref{thm:weight-monotone} to parametrize the search space for the WoT vector $\wot$ (defined in~\eqref{eq:wot-vector}) and solve the PAP~\eqref{eq:pap-specific} exactly. 

    \subsection*{Results and Discussion}
        Our first simulation evaluates the PoW difficulty level $\pow(\compPower)$ assigned to an agent. The optimal PoW difficulty level $\pow(\compPower)$ for each agent vs. the number of agents $\totalAgents$ is plotted in Fig.~\ref{fig:pow}. For a fixed value of $\totalAgents$, $\pow(\compPower)$ increases with $\compPower$ (Theorem~\ref{thm:effort-monotone}). As $\totalAgents$ increases, $\pow(\compPower)$ increases for all agents as the transaction rate is directly proportional to $\totalAgents$. Also, $\pow(\compPower)$ is a concave function of $\totalAgents$.
    Therefore, the marginal increase in PoW difficulty level decreases with $\totalAgents$. This implies that our proposed transaction rate control mechanism does not incur a substantial increase in computation cost to agents with an increase in application size (number of agents). Therefore, the transaction rate control mechanism~\eqref{eq:pap-specific} is suitable for IoT applications.
	
	\def\arraystretch{1.5}
	\begin{table}
		\centering
		\caption{Simulation Parameters for Transaction Rate control Problem~\eqref{eq:pap-specific}}
		\begin{tabular}{r|r|l}
			\hline
			\textbf{Parameters} & \textbf{Eq.} & \textbf{Value}\\
			\hline
			$\compPowerSet$ & \eqref{eq:pap-specific-objective} & $\{1, 3, 10\}$\\
			$\powSet$ & \eqref{eq:pap-specific-objective} & $\{1,2,\ldots,12\}$\\
			$p=(p(\compPower))_{\compPower\in\compPowerSet}$ & \eqref{eq:pap-specific-objective} & $[1/3\,\, 1/3\,\, 1/3]$\\
			$\alpha$ & \eqref{eq:pap-specific-objective} & 0.1\\
			$\beta$ & \eqref{eq:pap-specific-incentive-constraint} & 80\\
			$u_0$ & \eqref{eq:pap-specific-participation-constraint} & 10\\
			\hline
		\end{tabular}
		\label{tab:model-parameters}
	\end{table}
	
	\begin{figure}
		\centering
		\includegraphics[scale=0.4]{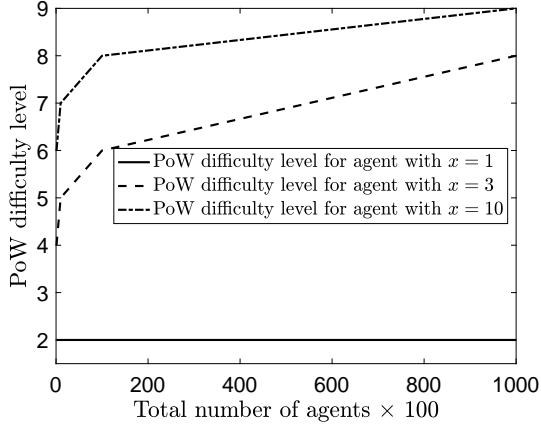}
		\caption{PoW difficulty level $\pow(\compPower)$ vs. number of agents $\totalAgents$ for our proposed transction rate control mechanism~\eqref{eq:pap-specific}. As $\totalAgents$ increases, $\pow(\compPower)$ increases for all agents but the
			marginal increase in PoW difficulty level decreases with $\totalAgents$. So, our proposed transaction rate control mechanism does not incur a substantial increase in computation cost to agents with an increase in the number of agents.}
		\label{fig:pow}
	\end{figure}
	
	Our next simulation evaluates the WoT $\wot(\compPower)$ assigned to an agent. Fig.~\ref{fig:weight} displays the variation of WoT $\wot(\compPower)$ for each agent vs.\ number of agents $\totalAgents$. 
 The simulation shows that $\wot(\compPower)$ increases with $\compPower$ (Theorem~\ref{thm:weight-monotone}). 
	Moreover, the WoT $\wot(\compPower)$ increases with $\totalAgents$ to compensate for the increase in PoW difficulty level. Also, the marginal increase in  WoT decreases with $\totalAgents$. Hence,  our proposed transaction rate control mechanism~\eqref{eq:pap-specific} ensures that the average approval time for agents with small computing power does not degrade rapidly with the number of agents.
	\begin{figure}
		\centering
		\includegraphics[scale=0.4]{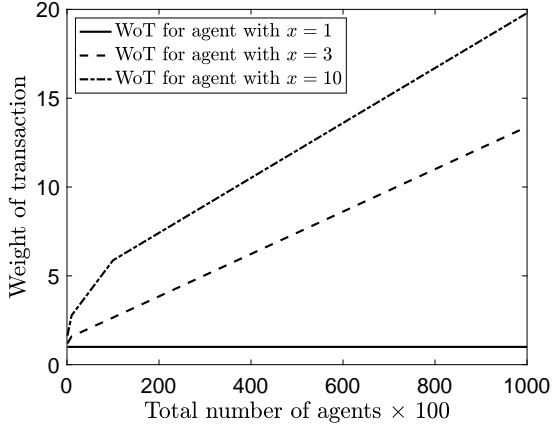}
		\caption{WoT $\wot(\compPower)$ vs. number of agents $\totalAgents$ for our proposed transaction rate control mechanism~\eqref{eq:pap-specific}. WoT $\wot(\compPower)$ increases with $\totalAgents$ to compensate for the increase in PoW difficulty level. Marginal increase in WoT decreases with $\totalAgents$. Hence, our proposed transaction rate control mechanism~\eqref{eq:pap-specific} ensures that the relative WoT of agents with small computing power (with respect to the WoT of agents with high computing power) does not degrade rapidly with the number of agents.}
		\label{fig:weight}
	\end{figure}
	
	We now use the optimal PoW difficulty level $\pow^*(\compPower)$ and optimal WoT $\wot^*(\compPower)$, obtained from~\eqref{eq:pap-specific}, to simulate the dynamics of Tangle. Tangle protocol and its evolution have been described in Sec.~\ref{sec:tangle-protocol}. At each time $t$, new transactions are added by each agent at a rate that depends on their computing power and the PoW difficulty level. Each new transaction chooses two tip transactions randomly for approval; the probability that a tip transaction is selected for approval is proportional to its weight. As different transactions have different weights, we use the accept-reject method~\cite{2016:VK} to select tip transactions non-uniformly during simulations. An important parameter associated with the dynamics of Tangle is the average approval time of transactions. The average approval time of a transaction is defined as the time between when a transaction is added to the ledger and when it gets approved by a new transaction. The average approval time of a transaction decreases with the WoT and increases with the transaction rate per number of agents\footnote{As the number of agents increases, so does the number of tip transactions; therefore, the approval time is an increasing function of the transaction rate divided by the number of agents.}. The probability that a tip transaction is selected for approval at time $t$ is proportional to the WoT. Therefore, the more the WoT, the higher its chance of being selected for approval. Also, if the PoW difficulty level increases, then the rate of new transactions decreases. Consequently, the number of tip transactions selected for approval also decreases. This leads to an increase in average approval time.
	\begin{figure}
		\centering
		\includegraphics[scale=0.4]{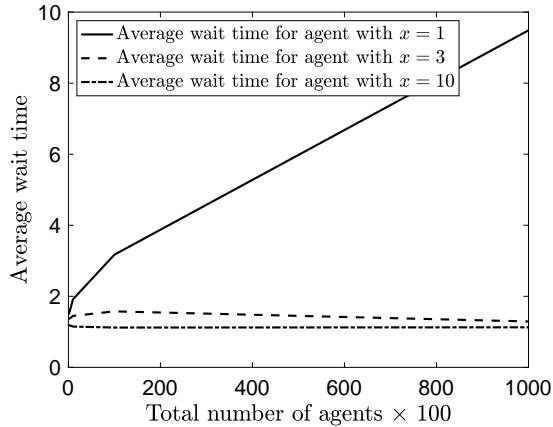}
		\caption{Average approval time of transactions vs. the number of agents $\totalAgents$ for our proposed transaction rate control mechanism solved by the PAP~\eqref{eq:pap-specific}. Agents assigned difficult PoW have to wait for a lesser time on average for approval of their tip transactions. This is because they are assigned a higher WoT. Agents assigned lower PoW difficulty levels have to wait for longer on average for approval of their transactions. This is because they are assigned a lower WoT.
  Compared to the fixed transaction rate control scheme in~\cite{2018:SP}, our proposed transaction rate control mechanism~\eqref{eq:pap-specific} achieves two benefits: (1) it moderates the rate of new transactions and increases the security of Tangle by incentivizing agents with high computing power to perform difficult PoW, (2) it allows agents with lower computing power to perform easier PoW at the expense of larger average approval time of transactions.
  }
		\label{fig:wait_time}
	\end{figure}
	
 Fig.~\ref{fig:wait_time} plots the average time for approval of transactions for each agent vs. the number of agents $\totalAgents$. As $\wot(1),\pow(1)$ remains unchanged with $\totalAgents$ (refer Fig.~\ref{fig:pow} and Fig.~\ref{fig:weight}), the average approval time for the agent with the lowest computing power increase with $\totalAgents$. This is because both the relative WoT (with respect to other agents) decreases, and so does the transaction rate per number of agents. Agents that are assigned difficult PoW have to wait for a lesser amount of time on average for approval of their tip transactions than agents that are assigned easier PoW. As observed in Fig.~\ref{fig:wait_time}, the average approval time for agents with higher computing power is almost constant because a decrease in the transaction rate per number of agents offsets an increase in relative WoT. 
 Compared to the fixed transaction rate control
scheme in~\cite{2018:SP}, our proposed transaction rate control mechanism~\eqref{eq:pap-specific} achieves two benefits: (1) it moderates
the rate of new transactions and increases the security of Tangle by incentivizing agents with high computing
power to perform difficult PoW, (2) it allows agents with lower computing power to perform easier PoW at
the expense of a larger average approval time of transactions.
	
	\begin{figure}
		\centering
		\includegraphics[scale=0.4]{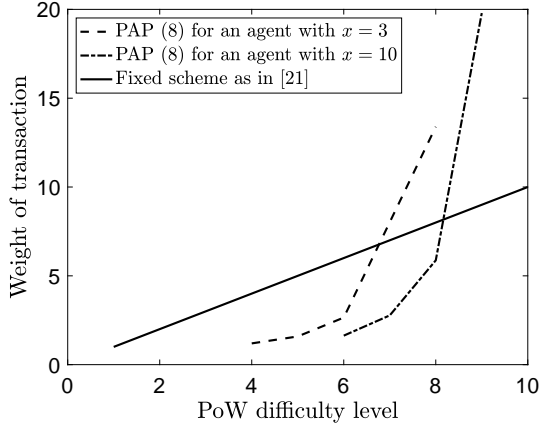}
		\caption{Comparison of WoT vs. PoW assigned to agents by our proposed transaction rate control mechanism solved by the PAP~\eqref{eq:pap-specific}, and the fixed transaction control scheme as in~\cite{2018:SP}. 
  The fixed scheme in \cite{2018:SP} assigns WoT as a linear function of the PoW difficulty level; this does not guarantee that agents will perform a difficult PoW. Our transaction rate control mechanism~\eqref{eq:pap-specific} assigns WoT as a convex function of PoW difficulty level; this incentivizes agents to perform difficult PoW. Transactions with higher PoW difficulty level makes it
difficult for an adversary to tamper transactions in Tangle; this makes Tangle more secure.
  } 
  \label{fig:comparison}
	\end{figure}
    	We now compare the proposed transaction rate control mechanism~\eqref{eq:pap-specific} with the fixed transaction control scheme as in~\cite{2018:SP}: it allocates the WoT as a fixed linear function of the PoW difficulty level. 
     Fig.~\ref{fig:comparison} plots WoT vs. PoW for agents with computing power $x=3$ and $x=10$ obtained from PAP~\eqref{eq:pap-specific}; it is compared with the fixed transaction control scheme as in~\cite{2018:SP}. The fixed scheme in~\cite{2018:SP} assigns WoT as a linear function of the PoW difficulty level; this does not guarantee that agents will perform a difficult PoW. Our transaction rate control mechanism~\eqref{eq:pap-specific} assigns WoT as a convex function of PoW difficulty level; this incentivizes agents to perform difficult PoW. Transactions with higher PoW difficulty level makes it
difficult for an adversary to tamper with transactions in Tangle; this makes Tangle more secure.
  Hence, to incentivize difficult PoW, the marginal increase in WoT should increase with the PoW difficulty level.
     
	To summarize, we simulated the PAP~\eqref{eq:pap-specific} for controlling the transaction rate in Tangle. As the PAP~\eqref{eq:pap-specific} is a mixed-integer program, we obtained the optimal solution by solving a set of linear programs. Theorem~\ref{thm:effort-monotone} was exploited to reduce the number of linear programs to be solved. We also simulated the dynamics of Tangle after incorporating the transaction rate control mechanism~\eqref{eq:pap-specific} and compared it with the fixed transaction rate control scheme in~\cite{2018:SP}.
	
	\section{Conclusion and Future Work}
	\label{sec:conclusion}
	Tangle is a distributed ledger technology suitable for IoT applications. Motivated by designing strategic contracts with partial information in microeconomics, this paper has proposed a principal-agent problem (PAP) approach to transaction rate control in Tangle. The principal (transaction rate controller) designs a mechanism to assign proof of work (PoW) difficulty level and weight of transaction (WoT) to agents (IoT devices). As the principal cannot observe agents' state (computing power), the principal also has to incentivize agents to be truthful. Our main results regarding the proposed transaction rate controller were the following: 1) the optimal PoW difficulty level increases with the computing power of the agents; 2) the optimal WoT increases with the computing power of the agent. We also simulated the dynamics of Tangle using the solution obtained from our proposed transaction rate control mechanism. We observed that agents with higher computing power are assigned higher PoW difficulty levels but have a smaller average transaction approval time than agents with lower computing power. As the mechanism is truth-telling, it incentivizes agents with higher computing power to perform a difficult PoW; this makes Tangle more secure. Finally, we compared the proposed mechanism with the transaction rate control scheme from the white paper on Tangle~\cite{2018:SP}; we observed that a concave relation between WoT and PoW is required to incentivize the agents to be truthful.
	
	Our transaction control mechanism ensures that agents are truth-telling. Still, it does not prevent agents from colluding, i.e.,  it may be advantageous for multiple agents to combine their computing power and act as a single agent. It would, therefore, be interesting to study the rate control problem that disincentivizes agents from forming coalitions. This is essential from the security viewpoint, as collusion of agents can lead to a majority attack on distributed ledgers.
	
	\begin{acks}
		This research was supported in part by the U.S.\ Army Research Office grant  
		W911NF-21-1-0093, National Science Foundation grant CCF-2112457, and Air Force Office of Scientific Research grant FA9550-22-1-0016.
	\end{acks}
	
	\bibliographystyle{ACM-Reference-Format}
	\bibliography{anurag}

\end{document}